\documentclass[12pt]{amsart}

\usepackage{amsmath,amsfonts,amssymb,amsthm,graphicx,cite}
\usepackage{mathrsfs}

\begin{document}

\theoremstyle{plain}
\newtheorem{theorem}{Theorem}
\newtheorem{corollary}{Corollary}
\newtheorem{lemma}{Lemma}
\newtheorem{proposition}{Proposition}

\theoremstyle{remark}
\newtheorem{remark}{Remark}

\newcommand{\adiffop}{A$\Delta$O}
\newcommand{\adiffops}{A$\Delta$Os}

\newcommand{\be}{\begin{equation}}
\newcommand{\ee}{\end{equation}}

\newcommand{\Z}{{\mathbb Z}}
\newcommand{\N}{{\mathbb N}}
\newcommand{\C}{{\mathbb C}}
\newcommand{\Cs}{{\mathbb C}^{*}}
\newcommand{\R}{{\mathbb R}}
\newcommand{\Q}{{\mathbb Q}}

\newcommand{\cK}{{\mathcal K}}
\newcommand{\cC}{{\mathcal C}}
\newcommand{\cF}{{\mathcal F}}

\newcommand{\re}{{\rm Re}\, }
\newcommand{\im}{{\rm Im}\, }

\newcommand{\Ln}{{\rm Ln}}
\newcommand{\Arg}{{\rm Arg}}

\newcommand{\Ai}{{\rm Ai}}
\newcommand{\Bi}{{\rm Bi}}

\title[A Product formula for a quartic oscillator]{A product formula for the eigenfunctions of a quartic oscillator}

\author{Martin Halln\"as}
\address{Department of Mathematical Sciences, Loughborough University, Leicestershire LE11 3TU, UK}
\email{m.a.hallnas@lboro.ac.uk}

\author{Edwin Langmann}
\address{Department of Theoretical Physics, KTH Royal Institute of Technology, SE-106 91 Stockholm, Sweden}
\email{langmann@kth.se}

\dedicatory{In memory of Vadim Kuznetsov}

\begin{abstract}
We consider the Schr\"odinger operator on the real line with an even quartic potential. Our main result is a product formula of the type
\begin{equation*}
	\psi_k(x)\psi_k(y) = \int_\R \psi_k(z)\cK(x,y,z)dz
\end{equation*}
for its eigenfunctions $\psi_k$. The kernel function $\cK$ is given explicitly in terms of the Airy function $\Ai(x)$, and it is positive for appropriate parameter values. As an application, we obtain a particular asymptotic expansion of the eigenfunctions $\psi_k$.
\end{abstract}

\maketitle

\section{Introduction}
In this paper we shall be concerned with solutions to the eigenvalue problem consisting of the differential equation
\be\label{diffEq}
	H(a,\lambda)\psi\equiv -\frac{d^2\psi}{dx^2} + \left(ax^2 + \frac{\lambda}{2}x^4\right)\psi = E\psi,\quad x\in\R,
\ee
and boundary conditions
\be\label{bConds}
	\lim_{x\to\pm\infty}\psi(x) = 0.
\ee
Throughout the paper we shall assume that $a\in\R$ and $\lambda>0$. As is well known, this eigenvalue problem has a discrete spectrum consisting of real eigenvalues $E_0<E_1<\cdots<E_k<\cdots$ with $E_k\to\infty$ as $k\to\infty$, and all eigenspaces are one-dimensional; see, e.g., Berezin and Shubin \cite{BS91}.

For any (complex) value of $E$, the differential equation \eqref{diffEq} has a unique solution $\psi=\psi(a,\lambda,E;x)$ with asymptotic behaviour
\be\label{asymp}
	\psi \sim x^{-1}\exp\left(-\frac{(\lambda/2)^{1/2}}{3}x^3-\frac{a}{2(\lambda/2)^{1/2}}x\right),\quad x\to +\infty,
\ee
and the asymptotic behaviour of $\psi^\prime$ is given by the derivative of the right-hand side of \eqref{asymp}; see, e.g., Hsieh and Sibuya \cite{HS66} or the appendix by Dicke in \cite{Sim70}. The eigenvalues $E=E_k$ of \eqref{diffEq}--\eqref{bConds} are precisely those values of $E$ for which $\psi$ also decays as $x\to-\infty$, which happens if and only if $\psi$ is either even or odd. We let $\psi_k(a,\lambda;x)\equiv\psi(a,\lambda,E_k;x)$ denote the corresponding eigenfunctions normalised by \eqref{asymp}. We note that this choice of normalisation is motivated by the relative simplicity of the corresponding product formula. Since $\psi_k$ has $k$ real zeros, we have that $\psi_k$ is even (odd) if and only if $k$ is even (odd).

We shall have occasion to make use of the well known invariance of the differential equation \eqref{diffEq} under the scaling transformation
\be\label{scaling}
	(a,\lambda,E,x)\to (\beta^{-2/3}a,\beta^{-1}\lambda,\beta^{-1/3}E,\beta^{1/6}x),\quad \beta > 0.
\ee
Insisiting on the asymptotic behaviour \eqref{asymp}, this entails that the (decaying) solutions of \eqref{diffEq} transform according to
\be\label{psiScale}
	\psi(a,\lambda,E;x) = \beta^{1/6}\psi\big(\beta^{-2/3}a,\beta^{-1}\lambda,\beta^{-1/3}E;\beta^{1/6}x\big),
\ee
so that no generality would be lost by fixing $\lambda=1$, say, once and for all. However, we find it instructive not to do so, since many of our formulas depend in a rather non-obvious manner on $\lambda$.

We are now in a position to describe our results. Before doing so we find it worth emphasising that they do not have non-trivial limits as $\lambda\to 0$. In other words they are non-perturbative in the sense that the harmonic oscillator is not contained as a limiting case.

We note that any product $\psi_k(x)\psi_k(y)$ of two eigenfunctions corresponding to the same eigenvalue $E_k$ satisfies the partial differential equation
\be\label{PDE}
	\big(H(x) - H(y)\big)\Psi(x,y) = 0.
\ee
Since the eigenfunctions are either even or odd, such a product is invariant under the (reflection) group $G$ generated by the  interchange $x\leftrightarrow y$ and the simultaneous change of sign $(x,y)\to(-x,-y)$. It is thus natural to consider solutions of \eqref{PDE} in terms of the variables 
\be\label{variables}
	u = xy,\quad v = (x^2+y^2)/2,
\ee
which generate the polynomial invariants of $G$. Although it will be important for us to allow any (real) values for $x$ and $y$, we note that the mapping $(x,y)\to(u,v)$ is one-to-one when restricted to regions such as $-y<x<y$.

In Section \ref{Sec2} we show that \eqref{PDE} admits solutions by separation of variables in terms of the variables $(u,v)$. More specifically, these solutions consist of an elementary function of $u$ and a function of $v$ given in terms of a solution of Airy's differential equation. We note that (up to conjugation and change of variable) Airy's differential equation is of confluent hypergeometric type whereas the differential equation \eqref{diffEq} is of triconfluent Heun type. For our purposes, the key property of these solutions is the fact that they depend on three variables $(x,y,z)$, where $z$ is essentially the separation constant, in a symmetric manner.

In Section \ref{Sec3} we establish a product formula of the type
\begin{equation}\label{intProdForm}
	\psi_k(x)\psi_k(y) = \int_\R \psi_k(z)\cK(x,y,z)dz
\end{equation}
for the eigenfunctions $\psi_k$ of the eigenvalue problem \eqref{diffEq}--\eqref{bConds}; see Theorem \ref{Thm} for the precise statement. The kernel function $\cK(x,y,z)$ is one of the solutions of \eqref{PDE} obtained in Section \ref{Sec2}, given in terms of the standard solution $\Ai(x)$ of Airy's equation.

Such product formulas are useful in various contexts. For example, they can be used to define a convolution product that is important for the harmonic analysis of the corresponding eigenfunction expansions. Since the role played by product formulas in harmonic analysis is beyond the scope of this paper (and indeed somewhat beyond our expertise), we have not attempted to provide a comprehensive list of relevant references. Instead, we would like to mention two sources of inspiration that were particularly important for us when writing this paper, and which contain numerous references to the wider literature. One is a paper by Connett et al.~\cite{CMS93} on product formulas and convolutions for angular and radial spheroidal wave functions, and the other is work of Koornwinder and collaborators on Jacobi functions and polynomials, as presented, e.g., in Koornwinder's paper \cite{Kor84}.

A product formula of the form \eqref{intProdForm} can also be regarded as an integral equation for the eigenfunctions $\psi_k$ after a suitable choice of the variable $x$, say. At this point we should mention that integral equations for functions from the Heun class, as well as kernel functions in the sense of solutions of a corresponding PDE of the form \eqref{PDE}, has a long history. For example, an integral equation satisfied by Lam\'e functions was established in a paper by Whittaker \cite{Whi15} published in 1915; see also Section 23.6 of Whittaker and Watson \cite{WW35}. More details and references to the rather extensive literature on the subject can, e.g., be found in Section 31.10 of \cite{Dig10} and in a book by Slavyanov and Lay \cite{SL00}. Of particular relevance to this paper is a paper by Kazakov and Slavyanov \cite{KS96}, which contains numerous kernel functions and integral equations for functions of Heun type as well as confluent cases thereof. Specifically, an integral equation we obtain for the eigenfunctions $\psi_k$ with $k$ even is closely related to an integral equation given in Theorem 5 in {\it loc.~cit.}; see Section \ref{Sec5} for further details.

As an application of the product formula, we deduce in Section \ref{Sec4} an asymptotic expansion of the eigenfunctions $\psi_k$ with respect to a particular (asymptotic) sequence of functions given by integrals involving the Airy function $\Ai(x)$.

In Section \ref{Sec5} we present a brief outlook on possible avenues for future research that we find particularly interesting.

Finally, in Appendix \ref{AppA} we collect properties of $\Ai(x)$ we have occasion to use in the main text, and in Appendix \ref{AppB} we deduce bounds on the kernel function that is used in Section \ref{Sec3}.

\medskip

\noindent {\bf Note.} Two and a half years after completing this paper we learnt from T.~T.~Truong that he obtained a product formula for the quartic oscillator already in 1974. More precisely, in his papers \cite{Tru74,Tru75} Truong showed that the orthonormal eigenfunctions $\psi_k(x)$ of the Schr\"odinger operator
\be
H\equiv \frac{1}{2}\left(-\frac{d^2}{dx^2}+\omega^2x^2+x^4\right)
\ee
satisfy a product formula
\be
\psi_k(x)\psi_k(y)=\mu_k\int_{-\infty}^\infty \psi_k(z){}^{\cosh xyz}_{\sinh xyz} \Ai\left(\frac{\omega^2+x^2+y^2+z^2}{2^{2/3}}\right)dz
\ee
for some constants $\mu_k$, where one should choose $\cosh$ and $\sinh$ for even and odd eigenfunctions $\psi_k$, respectively. In addition, adding a potential term $L(L+1)/2x^2$, he obtained a similar product formula for the corresponding eigenfunctions, in which $\cosh$ or $\sinh$ is replaced by a (modified) Bessel function. Moreover, in the recent interesting paper \cite{Tru16}, he used his earlier results to establish new integral properties of the quartic oscillator eigenfunctions. We would like to thank T.~T.~Truong for making us aware of his work on the quartic oscillator and for providing us with copies of his papers.

\section{Kernel functions}\label{Sec2}
In this section we demonstrate that the partial differential equation \eqref{PDE} admits solutions by separation of variables in terms of the variables $(u,v)$ given by \eqref{variables}.

To this end we observe that
\be
	H(x) - H(y) = (x^2-y^2)\left(\frac{\partial^2}{\partial u^2} - \frac{\partial^2}{\partial v^2} + a + \lambda v\right).
\ee
Hence, we have a solution of the form $\Psi(x,y)=\Phi(u(x,y),v(x,y))$ of equation \eqref{PDE} if the function $\Phi(u,v)$ satisfies the partial differential equation
\be
	\left(\frac{\partial^2}{\partial u^2} - \frac{\partial^2}{\partial v^2} + a + \lambda v\right)\Phi(u,v) = 0.
\ee
Moreover, assuming that $\Phi(u,v)=f(u)g(v)$, it is clear that we obtain a solution of the latter equation by requiring that the functions $f$ and $g$ satisfy the ordinary differential equations
\be\label{ODEs}
	\frac{d^2 f}{du^2} = \kappa^2 f,\quad \frac{d^2 g}{dv^2} = (a + \lambda v + \kappa^2)g,
\ee
for some separation constant $\kappa^2$. It follows that $f(u)$ should be a linear combination of $e^{\kappa u}$ and $e^{-\kappa u}$. Furthermore, introducing a function $w$ by requiring that
\be
	g(v)=w(\lambda^{1/3}v+(\kappa^2+a)/\lambda^{2/3}),
\ee
it is readily verified that $w(x)$ should satisfy Airy's differential equation 
\be\label{AiryEq}
	\frac{d^2 w}{dx^2} = xw.
\ee
Introducing a new variable $z$ by setting $\kappa=(\lambda/2)^{1/2}z$ and reverting to the variables $x$ and $y$, we find that the solutions $f$ and $g$ of the equations \eqref{ODEs} are invariant under any permutations of the three variables $(x,y,z)$. We thus have the following result.

\begin{proposition}\label{kernelProp}
Let $w(x)$ be a solution of Airy's equation \eqref{AiryEq}. Then the function
\begin{multline}\label{cK}
	\cK(a,\lambda;x,y,z)\\ \equiv \exp\big((\lambda/2)^{1/2}xyz\big)w\left(\frac{\lambda^{1/3}}{2}(x^2+y^2+z^2)+\frac{a}{\lambda^{2/3}}\right)
\end{multline}
satisfies the identities
\be\label{Kids}
	H(x)\cK(x,y,z) = H(y)\cK(x,y,z) = H(z)\cK(x,y,z).
\ee
\end{proposition}

Of course the identities \eqref{Kids} remain valid under multiplication of the function $\cK$ by any constant. In the next section we shall make use of this freedom to choose a particularly convenient kernel function.

\section{Product formula}\label{Sec3}
In this section we establish the product formula for the eigenfunctions $\psi_k$, $k=0,1,\ldots$, of the eigenvalue problem \eqref{diffEq}--\eqref{bConds}. As before, we normalise the eigenfunctions by insisting on the asymptotic behaviour \eqref{asymp}. The relevant kernel function is obtained from \eqref{cK} after multiplication by $\lambda^{1/3}$ and setting
\be
	w(x) = \Ai(x),\quad x\in\R,
\ee
where the standard solution $\Ai(x)$ of Airy's equation \eqref{AiryEq} can, in particular, be defined by its integral representation \eqref{intRep}. For a detailed account of this and other standard Airy functions see, e.g., Chapter 9 in \cite{Dig10} and references therein. The precise result can now be formulated as follows.

\begin{theorem}\label{Thm}
For $a\in\R$, $\lambda>0$ and $k=0,1,\ldots$, we have a product formula
\be\label{prodForm}
	\psi_k(x)\psi_k(y) = \int_\R \psi_k(z)\cK(x,y,z)dz
\ee
with kernel function
\begin{multline}\label{cKThm}
	\cK(a,\lambda;x,y,z)\\ \equiv \lambda^{1/3}\exp\big((\lambda/2)^{1/2}xyz\big)\Ai\left(\frac{\lambda^{1/3}}{2}(x^2+y^2+z^2)+\frac{a}{\lambda^{2/3}}\right).
\end{multline}
Moreover, as long as $a>\lambda^{2/3}a_1$ with $a_1=-2.3381074105\ldots$ being the first zero of $\Ai(x)$, the kernel function $\cK(x,y,z)$ is positive for all $x,y,z\in\R$.
\end{theorem}

\begin{proof}
We note that $\cK(a,\lambda;x,y,z)>0$ for $a>\lambda^{2/3}a_1$ since $\Ai(x)>0$ for $x>a_1$, see, e.g., Chapter 9 in \cite{Dig10}; and that convergence of the integral in \eqref{prodForm} is clear from the asymptotic behaviour of $\psi_k(z)$ and $\Ai(x)$ (c.f.~\eqref{asymp} and \eqref{AiAsymp}, respectively).

Having insisted on the asymptotic behaviour \eqref{asymp}, the eigenfunctions $\psi_k$ transform according to \eqref{psiScale} under the scaling transformation \eqref{scaling}. (Note that the statement of the Theorem is consistent with this transformation property.) Hence, it is sufficient to establish the product formula \eqref{prodForm} for $\lambda=8$, in which case the kernel function takes the particularly simple form
\be\label{kernel}
	\cK(a;x,y,z) = 2\exp(2xyz)\Ai(x^2+y^2+z^2+b),\quad b=a/4.
\ee

We continue by showing that the right-hand side of \eqref{prodForm} is a solution to the eigenvalue problem \eqref{diffEq}--\eqref{bConds} with $E=E_k$. Starting with the differential equation \eqref{diffEq}, we should establish that
\be\label{intEigEq}
	H(x)\int_\R \psi_k(z)\cK(x,y,z)dz = E_k\int_\R \psi_k(z)\cK(x,y,z)dz.
\ee
Assuming for now that we are allowed to differentiate twice under the integral sign, we can proceed as follows. First, acting with $H(x)$ on the kernel function and making use of the identities \eqref{Kids}, we find that the left-hand side of \eqref{intEigEq} equals
\be
	\int_\R \psi_k(z)H(z)\cK(x,y,z)dz.
\ee
Second, integrating by parts twice, we transfer the action of $H(z)$ to the factor $\psi_k(z)$. Third, invoking the eigenfunction property $H\psi_k=E_k\psi_k$, we arrive at the right-hand side of \eqref{intEigEq}.

These formal arguments are readily justified by the bounds on the kernel function $\cK(x,y,z)$, as well as derivatives thereof, that is provided by Lemma \ref{bLemma} in Appendix \ref{AppB}. Indeed, when combined with the asymptotic behaviour of $\psi_k(z)$, the bound \eqref{cKnBound} entails that the $n$th-order $x$-derivative of the integrand in \eqref{prodForm} is uniformly continuous for $(x,z)\in\R^2$. Consequently, we may differentiate any number of times under the integral sign. Furthermore, since $\cK(x,y,z)$ is invariant under the interchange $x\leftrightarrow z$, the bound holds true also for derivatives with respect to $z$. Taking into account the asymptotic behaviour of $\psi_k(z)$ and $\psi_k^\prime(z)$, it follows that our use of integration by parts is allowed.

To verify that the right-hand side of \eqref{prodForm} satisfies the boundary conditions \eqref{bConds} we fix $y\in\R$ and $\rho\in(0,1)$. By the Cauchy-Schwartz inequality and Lemma \ref{bLemma} from Appendix \ref{AppB} there exists a positive constant $C$ such that
\begin{multline}
	\left|\int_\R \psi_k(z)\cK(x,y,z)dz\right|\\ < C||\psi_k||\left(\int_\R \exp\left(-\frac{2\rho}{3}(x^2+z^2)^{3/2}\right)dz\right)^{1/2}.
\end{multline}
Making use of the elementary estimate $(x^2+z^2)^{3/2}\geq (|x|^3+|z|^3)/2$, we can further bound the right-hand-side by
\be
	C||\psi_k||\exp\left(-\frac{\rho}{6}|x|^3\right)\left(\int_\R \exp\left(-\frac{\rho}{3}|z|^3\right)dz\right)^{1/2},
\ee
which clearly decays to zero as $x\to\pm\infty$. Since all eigenspaces of \eqref{diffEq}--\eqref{bConds} are one-dimensional and the kernel function $\cK(x,y,z)$ is invariant under the interchange $x\leftrightarrow y$, we can thus conclude that
\be\label{prodFormWithConsts}
	\int_\R \psi_k(z)\cK(x,y,z)dz = c_k\psi_k(x)\psi_k(y)
\ee
for some constants $c_k$, which remain to be determined. To this end we shall distinguish between even and odd eigenfunctions $\psi_k$ or equivalently between $k=2m$ and $k=2m+1$ for some $m=0,1,\ldots$

Starting with the even cases, we have $\psi_{2m}(0)\neq 0$, so that we can fix $x=0$. This yields (c.f.~\eqref{kernel})
\be\label{evenProdForm}
	4\int_0^\infty \psi_{2m}(z)\Ai(y^2+z^2+b)dz = c_{2m}\psi_{2m}(0)\psi_{2m}(y).
\ee
We proceed to compute the leading asymptotic behaviour of the left-hand side as $y\to\infty$. Integrating by parts, we rewrite it as
\begin{multline}\label{twoTerms}
	4\psi_{2m}(0)\int_0^\infty \Ai(y^2+s^2+b)ds\\ + 4\int_0^\infty \psi^\prime_{2m}(z)\int_z^\infty \Ai(y^2+s^2+b)dsdz.
\end{multline}
Changing integration variable to $t=s^2$, we find that the first term is given by
\be
	2\psi_{2m}(0)\int_0^\infty \Ai(t+y^2+b)t^{-1/2}dt.
\ee
Invoking the case $n=0$ of Lemma \ref{AiIntLemma} in Appendix \ref{AppA} and making use of the fact that $\Gamma(1/2)=\sqrt{\pi}$, we deduce that this term is asymptotic to $\psi_{2m}(0)\psi_{2m}(y)$ as $y\to\infty$.

Integrating by parts once more while keeping in mind that $\psi_{2m}^\prime(0)=0$, we rewrite the second term in \eqref{twoTerms} as
\be
	4\int_0^\infty \psi_{2m}^{\prime\prime}(z)\int_z^\infty\int_{s_2}^\infty \Ai(y^2+s_1^2+b)ds_1ds_2dz.
\ee
Since $\Ai(x)>0$ for $x>a_1$, we can, for sufficiently large $y$, bound the modulus of this term by
\be
	4||\psi_{2m}^{\prime\prime}||_\infty \int_0^\infty\int_z^\infty\int_{s_2}^\infty \Ai(y^2+s_1^2+b)ds_1ds_2dz.
\ee
Note that the fact that $\psi_{2m}$ is a solution of the differential equation \eqref{diffEq} for $E=E_{2m}$ with asymptotic behaviour \eqref{asymp} ensures that $||\psi_{2m}^{\prime\prime}||_\infty<\infty$. Reversing the order of integration, we can carry out the integration over $z$ and $s_2$, and thus reduce the triple integral to
\be
	\frac{1}{4}\int_0^\infty \Ai(t+y^2+b)t^{1/2}dt,
\ee
where we have again changed integration variable to $t=s_1^2$. Hence, Lemma \ref{AiIntLemma} with $n=1$ entails that the second term in \eqref{twoTerms} does not contribute to the leading asymptotic behaviour of \eqref{evenProdForm}. It follows that $c_{2m}=1$ for each $m=0,1,\ldots$

In the odd cases we have $\psi_{2m+1}(0)=0$.  On the other hand $\psi_{2m+1}^\prime(0)\neq 0$, so that a suitable starting point is
\be\label{oddProdForm}
	8y\int_0^\infty \psi_{2m+1}(z) z\Ai(y^2+z^2+b)dz = c_{2m+1}\psi_{2m+1}^\prime(0)\psi_{2m+1}(y),
\ee
obtained from \eqref{prodFormWithConsts} by first differentiating both sides with respect to $x$ and then setting $x=0$. Integrating by parts twice, we rewrite the left-hand side as
\begin{multline}\label{twoTerms2}
	8y\psi_{2m+1}^\prime(0)\int_0^\infty \int_{s_2}^\infty s_1\Ai(y^2+s_1^2+b)ds_1ds_2\\ + 8y\int_0^\infty \psi_{2m+1}^{\prime\prime}(z)\int_z^\infty\int_{s_2}^\infty s_1\Ai(y^2+s_1^2+b)ds_1ds_2dz.
\end{multline}
Just as we did above, we shall consider the two terms separately. Interchanging the order of integration, carrying out the integration over $s_2$, and then setting $t=s_1^2$, we deduce that the first term is equal to
\be
	4y\psi_{2m+1}^\prime(0)\int_0^\infty \Ai(t+y^2+b)t^{1/2}dt,
\ee
which by Lemma \ref{AiIntLemma} is asymptotic to $\psi_{2m+1}^\prime(0)\psi_{2m+1}(y)$. Moreover, following the reasoning in the even case, we readily bound the modulus of the second term by
\be
	\frac{2}{3}y ||\psi_{2m+1}^{\prime\prime\prime}||_\infty \int_0^\infty \Ai(t+y^2+b)t^{3/2}dt
\ee
for sufficiently large $y$. Invoking again Lemma \ref{AiIntLemma}, we can thus conclude that also $c_{2m+1}=1$ for each $m=0,1,\ldots$ This concludes the proof of Theorem \ref{Thm}.
\end{proof}

As a straightforward corollary of Theorem \ref{Thm}, we now infer the eigenfunction expansion of the kernel function $\cK(x,y,z)$ defined by \eqref{cKThm}. It is clear from Lemma \ref{bLemma} in Appendix \ref{AppB} that $\cK(x,y,z)$ is square-integrable in $z$, say. Since the eigenfunctions in question form a complete orthogonal set in $L^2(\R,dz)$ (see, e.g., Berezin and Shubin \cite{BS91}), it follows that
\be
	\cK(x,y,z) = \sum_{k=0}^\infty\frac{1}{||\psi_k||^2}\hat{\cK}_k(x,y)\psi_k(z)
\ee
with
\be
	\hat{\cK}_k(x,y)\equiv \int_\R \psi_k(z)\cK(x,y,z)dz.
\ee
Making use of the product formula \eqref{prodForm}, we thus arrive at the following result.

\begin{corollary}\label{Cor}
For $a\in\R$ and $\lambda>0$, we have the eigenfunction expansion
\be
	\cK(x,y,z) = \sum_{k=0}^\infty\frac{1}{||\psi_k||^2}\psi_k(x)\psi_k(y)\psi_k(z).
\ee
\end{corollary}

\section{An asymptotic expansion of the eigenfunctions}\label{Sec4}
We recall that the eigenfunctions $\psi_k$ admit an asymptotic expansion of the form
\be\label{psiFullAsymp}
	\psi_k(x)\sim x^{-1}\exp\left(-\frac{(\lambda/2)^{1/2}}{3}x^3-\frac{a}{2(\lambda/2)^{1/2}}x\right)\left(1 + \sum_{n=1}^\infty B_{k,n}x^{-n/2}\right)
\ee
as $x\to +\infty$; see Hsieh and Sibuya \cite{HS66}. In principle, the coefficients $B_{k,n}$ can be computed by substituting the right-hand side of \eqref{psiFullAsymp} for $\psi$ in the differential equation \eqref{diffEq}.

Elaborating on the latter part of the proof of Theorem \ref{Thm}, we shall in this section establish a rather different asymptotic expansion of the eigenfunctions. To simplify the computations and formulae involved we fix $\lambda=8$, which entails no loss of generality, c.f.~\eqref{scaling}--\eqref{psiScale}. Again, it will be important to consider even and odd eigenfunctions separately. 

We start with even case. For $x=0$, the product formula \eqref{prodForm} takes the simple, but non-trivial, form
\be\label{evenProdForm2}
	\psi_{2m}(0)\psi_{2m}(y) = 4\int_0^\infty \psi_{2m}(z)\Ai(y^2+z^2+b)dz,\quad b=a/4.
\ee
Successively integrating by parts, we find that the right-hand side can be rewritten as the series
\begin{multline}\label{intSeries}
	4\psi_{2m}(0)\int_0^\infty \Ai(s_1^2+y^2+b)ds_1\\ + 4\psi_{2m}^{(2)}(0)\int_0^\infty\int_{s_3}^\infty\int_{s_2}^\infty \Ai(s_1^2+y^2+b)ds_1ds_2ds_3\\ +\cdots + 4\psi_{2m}^{(2n)}(0)\int_0^\infty\int_{s_{2n+1}}^\infty\cdots\int_{s_2}^\infty \Ai(s_1^2+y^2+b)ds_1\cdots ds_{2n+1}\\ + R_{2n}(y),
\end{multline}
where the remainder term is given by
\begin{multline}\label{R2n}
	R_{2n}(y)\\ = 4\int_0^\infty \psi_{2m}^{(2n+1)}(z)\int_z^\infty\int_{s_{2n+1}}^\infty\cdots\int_{s_2}^\infty \Ai(s_1^2+y^2+b)ds_1\cdots ds_{2n+1}dz
\end{multline}
for each $n=0,1,\ldots$ Note that the odd terms in the series vanish, since $\psi_{2m}$ is even. Reversing the order of integration in each term, we can carry out all integrations in \eqref{intSeries} except those over $s_1$. Then introducing the variable $t=s_1^2$, we obtain the series
\be\label{intSeries2}
	2\psi_{2m}(0)\phi_0(y) + 2\frac{\psi_{2m}^{(2)}(0)}{2!}\phi_1(y) +\cdots + 2\frac{\psi_{2m}^{(2n)}(0)}{(2n)!}\phi_n(y) + R_{2n}(y),
\ee
with
\be
\phi_n(b;y)\equiv \int_0^\infty \Ai(t+y^2+b)t^{n-1/2}dt,\quad n=0,1,\ldots
\ee
We note that the integrals appearing on the right-hand side can be evaluated explicitly by combining the formula \eqref{AiryIntRed}, which reduces each integral to the special case $n=0$, with Aspnes' evaluation \eqref{AspnesInt} of the $n=0$ integral. Moreover, it is clear from Lemma \ref{AiIntLemma} in Appendix \ref{AppA} that the functions $\phi_n(y)$ form an asymptotic sequence in the sense that
\be
\phi_{n+1}(y)/\phi_n(y)\to 0,\quad y\to\infty,
\ee
for all $n=0,1,\ldots$ Following standard notation, we shall thus write
\be
f(y)\sim \sum_{n=0}^\infty a_n\phi_n(y)
\ee
if, for each $m=0,1,\ldots$, we have
\be
\left(f(y)-\sum_{n=0}^m a_n\phi_n(y)\right)\Bigg/\phi_m(y)\to 0,\quad y\to\infty.
\ee

Integrating by parts once more in \eqref{R2n}, and making use of the fact that $\psi_{2m}^{(2n+1)}(0)=0$, we deduce the following bound for the remainder term:
\begin{multline}
	|R_{2n}(y)|\\ < 4||\psi_{2m}^{(2n+2)}||_\infty \int_0^\infty\int_z^\infty\int_{s_{2n+2}}^\infty\cdots\int_{s_2}^\infty \Ai(s_1^2+y^2+b)ds_1\cdots ds_{2n+2}dz\\ = 2\frac{||\psi_{2m}^{(2n+2)}||_\infty}{(2n+2)!}\phi_{n+1}(y),
\end{multline}
valid for $y^2+b\geq a_1$. Note that the asymptotic behaviour of $\psi_{2m}$, combined with the fact that $\psi_{2m}$ satisfies the differential equation \eqref{diffEq} for $E=E_{2m}$, ensures that $||\psi_{2m}^{(2n+2)}||_\infty<\infty$.

Starting instead from the observation
\be
	\psi_{2m+1}^\prime(0)\psi_{2m+1}(y) = 8y\int_0^\infty \psi_{2m+1}(z)z\Ai(y^2+z^2+b)dz,
\ee
the odd case can be treated in a similar manner. Doing so we arrive at the following result.

\begin{proposition}\label{Prop}
For $a\in\R$ and $m=0,1,\ldots$, we have the asymptotic expansions
\be\label{psi2mAsymp}
	\psi_{2m}(y)\sim \frac{2}{\psi_{2m}(0)}\sum_{n=0}^\infty \frac{\psi_{2m}^{(2n)}(0)}{(2n)!}\phi_n(y)
\ee
and
\be\label{psi2m+1Asymp}
	\psi_{2m+1}(y)\sim \frac{4y}{\psi_{2m+1}^\prime(0)}\sum_{n=0}^\infty \frac{\psi_{2m+1}^{(2n+1)}(0)}{(2n+1)!}\phi_{n+1}(y)
\ee
as $y\to\infty$ (where $b=a/4$ and we have set $\lambda=8$).
\end{proposition}

We find it worth noting that one can also obtain the result in \eqref{psi2mAsymp} by substituting for $\psi_{2m}(z)$ in the right-hand side of \eqref{evenProdForm2} its Taylor series
\be
	\psi_{2m}(z) = \sum_{n=0}^\infty \frac{\psi_{2m}^{(2n)}(0)}{(2n)!}z^{2n}
\ee
and integration formally term by term. This process is readily justified by a suitable application of Lemma \ref{AiIntLemma} in Appendix \ref{AppA}, and is closely related to the use of Watson's lemma in the context of Laplace integrals; see, e.g., Copson \cite{Cop65}. This point of view highlights the fact that the specialisation \eqref{evenProdForm2} of the product formula connects representations of $\psi_{2m}$ around the origin and infinity, respectively. The result in \eqref{psi2m+1Asymp} can be obtained in a similar manner.

\section{Outlook}\label{Sec5}
In this paper our focus has been on establishing the product formula in Theorem \ref{Thm} for the eigenfunctions of the eigenvalue problem \eqref{diffEq}--\eqref{bConds} of a quartic oscillator. Below we add a few remarks on possible avenues for future research that we find particularly interesting.

Let us begin by considering the product formula as an integral equation for the eigenfunctions. More specifically, assuming $k=2m$ for some $m=0,1,\ldots$ and setting $x=0$ in \eqref{prodForm}, we arrive at
\begin{multline}\label{intEq}
	\psi_{2m}(0)\psi_{2m}(y)\\ = 2\lambda^{1/3}\int_0^\infty \psi_{2m}(z)\Ai\left(\frac{\lambda^{1/3}}{2}(y^2+z^2)+\frac{a}{\lambda^{2/3}}\right)dz.
\end{multline}
(In the odd cases $k=2m+1$ a non-trivial result is obtained by differentiating the product formula before setting $x=0$.) In Section \ref{Sec4} we exploited this point of view to deduce a particular asymptotic expansion of the eigenfunctions. One may ponder what further properties can be gleaned from the integral equation \eqref{intEq}. We note, in particular, that the eigenvalue $\psi_{2m}(0)$ provides the link between the asymptotic normalisation \eqref{asymp} adopted in this paper and the normalisation given by
\be
	\psi(0) = 1,\quad \psi^{\prime}(0) = 0.
\ee
It is also of interest that \eqref{intEq} is related to the integral equation appearing in Theorem 5 of Kazakov and Slavyanov \cite{KS96}. More precisely, setting $\lambda=8$ and $a=-4f$, writing $\psi_{2m}(z)=z^{1/2}w(z^2)$, and changing integration variable to $t=z^2$, we arrive at the integral equation (65) in {\it loc.~cit.} for $w$ with $\gamma=1/2$. (Note, however, that the condition $\gamma-1>0$ in said theorem is then not satisfied.)

In this paper we have only made use of the kernel function obtained from \eqref{cK} by setting $w(x)=\Ai(x)$. It is natural to enquire whether any use could be found for the kernel functions given by other solutions of Airy's equation \eqref{AiryEq}. For example, judging by its asymptotic behaviour, one may expect that the Airy function $\Bi(x)$ is relevant to solutions of the differential equation \eqref{diffEq} that tend to infinity as $x\to\pm\infty$.

Another interesting possibility is to attempt to make use of the product formula \eqref{prodForm} in the harmonic analysis of expansions in the eigenfunctions $\psi_k$.

\section*{Acknowledgements}
We would like to acknowledge a discussion with Vadim Kuznetsov on 16th September 2005 at KTH in which he conjectured that the quartic oscillator Hamiltonian should have a simple and explicit kernel function depending on three variables in a symmetric manner. We are grateful to Simon Ruijsenaars for pointing out to us that Vadim's conjecture was likely motivated by the existence of such kernel functions for the Heun equation, and the fact that the Schr\"odinger equation for the quartic oscillator is equivalent to a triconfluent Heun equation.

We thank Tom Koornwinder for explaining to us the significance of product formulas in harmonic analysis and for his encouragement to write this paper. In addition, we are grateful to Boris Shapiro for helpful remarks on literature related to the quartic oscillator, and M.~H.~is grateful to the participants in the Mathematical Physics seminar at Loughborough University for their many helpful comments.

We also acknowledge the comments of the anonymous referee, which helped us to improve the presentation.

This work was supported by the G\"oran Gustafsson Foundation and the Swedish Research Council (VR) under contract no.~621-2010-3708.

\begin{appendix}

\section{A standard solution of Airy's equation}\label{AppA}
In this appendix we collect properties of the standard solution $\Ai(x)$ of Airy's equation \eqref{AiryEq} we have occasion to use in the main text. For further details see, e.g., Chapter 9 in \cite{Dig10} and references therein.

First of all, we recall the integral representation
\be\label{intRep}
	\Ai(x) = \frac{1}{2\pi i}\int_{\infty e^{-i\pi/3}}^{\infty e^{i\pi/3}}\exp(\zeta^3/3-\zeta x)d\zeta,
\ee
where the contour of integration consists of two rays emerging from the origin at angles $\pm\pi/3$; and the asymptotic expansion
\be\label{AiAsymp}
	\Ai(x)\sim \frac{\exp\left(-\frac{2}{3}x^{3/2}\right)}{2\pi x^{1/4}}\sum_{k=0}^\infty \frac{\Gamma(3k+1/2)}{(-9)^k(2k)!}x^{-3k/2},
\ee
which is valid uniformly as $|x|\to\infty$ in any closed subsector of the sector $|\arg x|<\pi$.

In both Section \ref{Sec3} and \ref{Sec4} the integrals
\be\label{AiryInt}
	\int_0^\infty \Ai(t+x)t^{n-1/2}dt,\quad n=0,1,\ldots
\ee
play an important role. Using the fact that $\Ai(x)$ is a solution of Airy's equation \eqref{AiryEq}, all such integrals can be reduced to the special case $n=0$. More precisely, assuming $n=1,2,\ldots$, we have
\be\label{intRed}
\begin{split}
	\int_0^\infty \Ai(t+x)t^{n-1/2}dt &= \int_0^\infty \Ai^{\prime\prime}(t+x)t^{n-3/2}dt\\ &\qquad - x\int_0^\infty \Ai(t+x)t^{n-3/2}dt\\ &= \left(\frac{d^2}{d x^2} - x\right)\int_0^\infty \Ai(t+x)t^{n-3/2}dt.
\end{split}	
\ee
Iterating this process $n$ times, we obtain
\be\label{AiryIntRed}
	\int_0^\infty \Ai(t+x)t^{n-1/2}dt = \left(\frac{d^2}{d x^2} - x\right)^n\int_0^\infty \Ai(t+x)t^{-1/2}dt.
\ee
Moreover, Aspnes \cite{Asp66} obtained the following integral evaluation:
\be\label{AspnesInt}
	\int_0^\infty \Ai(t+x)t^{-1/2}dt = 2^{2/3}\pi\Ai^2(x/2^{2/3});
\ee
see Equation (B17) and set the normalisation constant $N=\pi$. Combining \eqref{AiryIntRed} with \eqref{AspnesInt} yields an explicit evaluation of the integrals in \eqref{AiryInt}. This enables us, in particular, to determine the asymptotic behaviour of these integrals, which we make use of in both the proof of Theorem \ref{Thm} and in the discussion leading up to Proposition \ref{Prop}. The precise result now follows.

\begin{lemma}\label{AiIntLemma}
For $n=0,1,\ldots$, we have an asymptotic expansion of the form
\begin{multline}\label{AiIntAsymp}
	\int_0^\infty \Ai(t+x)t^{n-1/2}dt\\ \sim \frac{\Gamma(n+1/2)}{2\sqrt{\pi}}\frac{\exp\left(-\frac{2}{3}x^{3/2}\right)}{x^{(n+1)/2}}\left(1+\sum_{k=1}^\infty \frac{C_{n,k}}{x^{3k/2}}\right)
\end{multline}
as $|x|\to\infty$ in any closed subsector of the sector $|\arg x|<\pi$.
\end{lemma}

\begin{proof}
We shall prove the statement by induction on $n$. Substituting the asymptotic expansion \eqref{AiAsymp} for $\Ai(x)$ in the right-hand side of \eqref{AspnesInt} and using the fact that $\Gamma(1/2)=\sqrt{\pi}$, we arrive at \eqref{AiIntAsymp} for $n=0$ with the coefficients $C_{0,k}$ given by
\be\label{C0k}
	C_{0,k} = \frac{1}{\pi}\left(-\frac{2}{9}\right)^k\sum_{m=0}^k\frac{\Gamma(3m+1/2)\Gamma(3(k-m)+1/2)}{(2m)!(2k-2m)!}.
\ee

Assume that \eqref{AiIntAsymp} holds true up to some $n=0,1,\ldots$ Then appealing to \eqref{intRed} with $n\to n+1$, we infer
\begin{multline}\label{AiIntAsymp2}
	\int_0^\infty \Ai(t+x)t^{n+1/2}dt\\ \sim \frac{\Gamma(n+1/2)}{2\sqrt{\pi}}\left(\frac{d^2}{d x^2} - x\right)\frac{\exp\left(-\frac{2}{3}x^{3/2}\right)}{x^{(n+1)/2}}\left(1+\sum_{k=1}^\infty \frac{C_{n,k}}{x^{3k/2}}\right).
\end{multline}
To establish that the right-hand side is of the required form we observe
\begin{multline}
	\frac{x^{(n+2)/2}}{\exp\left(-\frac{2}{3}x^{3/2}\right)}\circ \left(\frac{d^2}{d x^2} - x\right)\circ \frac{\exp\left(-\frac{2}{3}x^{3/2}\right)}{x^{(n+1)/2}} - (n+1/2)\\ = x^{1/2}\frac{d^2}{dx^2} - \left(2x+\frac{n+1}{x^{1/2}}\right)\frac{d}{dx} + \frac{(n+1)(n+3)}{4x^{3/2}}.
\end{multline}
By a direct computation, we deduce
\begin{multline}
	\left(x^{1/2}\frac{d^2}{dx^2} - \left(2x+\frac{n+1}{x^{1/2}}\right)\frac{d}{dx}\right)\left(1+\sum_{k=1}^\infty \frac{C_{n,k}}{x^{3k/2}}\right)\\ = \frac{3C_{n,1}}{x^{3/2}} + \frac{1}{4}\sum_{k=2}^\infty \frac{12kC_{n,k}+(3k-3)(3k+2n+1)C_{n,k-1}}{x^{3k/2}}.
\end{multline}
Since $(n+1/2)\Gamma(n+1/2)=\Gamma(n+3/2)$, it follows that \eqref{AiIntAsymp2} indeed is of the form \eqref{AiIntAsymp} with $n\to n+1$. Moreover, combining the last two formulae with \eqref{C0k} we could, in principle, calculate the values of all coefficients $C_{n,k}$.

Finally, we note that term-wise differentiation of the asymptotic series is permited since the asymptotic expansion \eqref{AiAsymp}, and by inference the expansions \eqref{AiIntAsymp}, is valid uniformly in any closed subsector of the sector $|\arg x|<\pi$.
\end{proof}

\section{Bounds on a kernel function}\label{AppB}
In the proof of Theorem \ref{Thm}, as well as when inferring its Corollary \ref{Cor}, we require suitable bounds on the kernel function \eqref{cKThm} and some of its derivatives. For simplicity, we shall restrict our attention to the particularly simple kernel function $\cK(a;x,y,z)$ obtained by setting $\lambda=8$, c.f.~\eqref{kernel}. The bounds are, however, readily generalised to other values of $\lambda$ by suitable scalings of the variables and parameters involved. The following lemma contains the precise statements and proofs of the pertinent bounds.

\begin{lemma}\label{bLemma}
Let $a,y\in\R$. For each $n=0,1,\ldots$ and $\rho\in(0,1)$, there exists a positive constant $C=C_n(\rho,a,y)$ such that
\be\label{cKnBound}
	\left|\frac{\partial^n\cK(a;x,y,z)}{\partial x^n}\right| < C\exp\left(-\frac{2\rho}{3}(x^2+z^2)^{3/2}\right)
\ee
for all $x,z\in\R$.
\end{lemma}

\begin{proof}
Since the Airy function $\Ai(x)$, and consequently the kernel function $\cK(a;x,y,z)$, is a smooth function, it is sufficient to establish the decay bounds
\be\label{dBound}
	\left|\frac{\partial^n\cK(a;x,y,z)}{\partial x^n}\right| = O\left(\exp\left(-\frac{2\rho}{3}(x^2+z^2)^{3/2}\right)\right)
\ee
as $x^2+z^2\to\infty$.

To this end we introduce polar coordinates
\be
	x = r\cos \theta,\quad z = r\sin \theta,
\ee
and consider the functions
\be\label{cFn}
	\cF_n(r,\theta)\equiv \frac{\partial^n\cK(a;x(r,\theta),y,z(r,\theta))}{\partial x^n}.
\ee
We find it convenient to introduce the function
\be
	\xi(r)\equiv r^2+y^2+b,
\ee
since it will appear repeatedly throughout the proof, and, in most cases, to suppress its dependence on $r$.

We observe that
\be
	\cF_0(r,\theta) = 2\exp(yr^2\sin 2\theta)\Ai(\xi),
\ee
and that each function $\cF_n$, $n=1,2,\ldots$, is given in terms of $\cF_{n-1}$ by the formula
\be
	\cF_n = \cos \theta\frac{\partial\cF_{n-1}}{\partial r} - \frac{\sin \theta}{r}\frac{\partial\cF_{n-1}}{\partial \theta}.
\ee
Proceeding by induction on $n$, it is readily verified that
\be
	\cF_n = \exp(yr^2\sin 2\theta)\big(P_n\Ai(\xi) + Q_n\Ai^\prime(\xi)\big)
\ee
for some polynomials $P_n=P_n(r)$ and $Q_n=Q_n(r)$ of degree at most $2n$ and $2n-1$, respectively, with coefficients given by polynomials in $\cos \theta$ and $\sin \theta$. Hence, we can find positive constants $C_n=C_n(a,y)$ such that
\be\label{cFnBound}
	|\cF_n(r,\theta)| < C_n(1+r^{2n})\exp(yr^2)\big(|\Ai(\xi)| + |\Ai^\prime(\xi)|\big).
\ee

From the the asymptotic expansions of $\Ai(x)$ and $\Ai^\prime(x)$ (see \eqref{AiAsymp} for the former and Eq.~(9.7.6) in \cite{Dig10} for the latter), we infer the decay bounds
\be
	\Ai(\xi) = O\left(\xi^{-1/4}\exp\left(-\frac{2}{3}\xi^{3/2}\right)\right),
\ee
\be
	\Ai^\prime(\xi) = O\left(\xi^{1/4}\exp\left(-\frac{2}{3}\xi^{3/2}\right)\right),
\ee
as $\xi\to\infty$. Moreover, we clearly have
\be\label{rExp}
	\xi(r)^{3/2} = r^3 + O(r),\quad r\to\infty.
\ee
Combining \eqref{cFnBound}--\eqref{rExp}, we conclude that

\be
	\cF_n(r,\theta) = O\left(\exp\left(-\frac{2\rho}{3}r^3\right)\right),\quad r\to\infty,
\ee
for each $\rho\in(0,1)$. Reverting to the variables $(x,z)$ and keeping \eqref{cFn} in mind, we arrive at the decay bounds \eqref{dBound}.
\end{proof}

\end{appendix}

\bibliographystyle{amsalpha}

\begin{thebibliography}{CMS93}

\bibitem[Asp66]{Asp66} D.~E.~Aspnes, \emph{Electric-field effects on optical absorption near thresholds in solids}, Phys.~Rev.~{\bf 147} (1966), 554--566.

\bibitem[BS91]{BS91} F.~A.~Berezin and M.~A.~Shubin, \emph{The Schr\"odinger equation}, Kluwer Academic Publishers, Dordrecht, 1991.

\bibitem[CMS93]{CMS93} W.~C.~Connett, C.~Markett and A.~L.~Schwartz, \emph{Product formulas and convolutions for angular and radial spheroidal wave functions}, Trans.~Amer.~Math.~Soc.~{\bf 338} (1993), 695--710.

\bibitem[Cop65]{Cop65} E.~T.~Copson, \emph{Asymptotic expansions}, Cambridge University Press, Cambridge, 1965.

\bibitem[Dig10]{Dig10} Digital Library of Mathematical Functions, Release date 2010-05-07, National Institute of Standards and Technology, http://dlmf.nist.gov.

\bibitem[HS66]{HS66} P.-F.~Hsieh and Y.~Sibuya, \emph{On the asymptotic integration of second order linear ordinary differential equations with polynomial coefficients},  J.~Math.~Anal.~Appl.~{\bf 16} (1966), 84--103.

\bibitem[KS96]{KS96} A.~Ya.~Kazakov and S.~Yu.~Slavyanov, \emph{Integral equations for special functions of Heun class}, Methods Appl.~Anal.~{\bf 3} (1996), 447--456.

\bibitem[Kor84]{Kor84} T.~H.~Koornwinder, \emph{Jacobi functions and analysis on noncompact semisimple Lie groups}, Special functions: group theoretical aspects and applications, 1--85, Math.~Appl., Reidel, Dordrecht, 1984.

\bibitem[Sim70]{Sim70} B.~Simon, \emph{Coupling constant analyticity for the anharmonic oscillator}, Ann.~Phys.~{\bf 58} (1970), 76--136.

\bibitem[SL00]{SL00} S.~Yu.~Slavyanov and W.~Lay, \emph{Special functions. A unified theory based on singularities}, Oxford University Press, Oxford, 2000.

\bibitem[Tru74]{Tru74} T.~T.~Truong, \emph{New integral equation for the quartic anharmonic oscillator}, Lett.~Nuovo Cimento {\bf 9} (1974), 533--536.

\bibitem[Tru75]{Tru75} T.~T.~Truong, \emph{Weyl quantization of anharmonic oscillators}, J.~Math.~Phys.~{\bf 16} (1975), 1034--1042.

\bibitem[Tru16]{Tru16} T.~T.~Truong, \emph{Quartic anharmonic oscillator integral properties via the 2D-quantum free fall problem}, Far East J.~Appl.~Math.~{\bf 94} (2016), 455--490.

\bibitem[Whi15]{Whi15} E.~T.~Whittaker, \emph{On Lam\'e's differential equation and ellipsoidal harmonics}, Proc.~Lond.~Math.~Soc.~{\bf 14} (1915), 260--268.

\bibitem[WW35]{WW35} E.~T.~Whittaker and G.~N.~Watson, \emph{A course of modern analysis}, 4th ed., Cambridge University Press, Cambridge, 1935.

\end{thebibliography}

\end{document}